\documentclass[a4paper,10.5pt]{amsart}%
\usepackage{amsfonts}
\usepackage{amssymb}
\usepackage{amssymb}
\usepackage{amscd}
\usepackage[dvips]{graphicx}
\usepackage{amsmath}%
\providecommand{\U}[1]{\protect\rule{.1in}{.1in}}
\newtheorem{lemma}{Lemma}
\newtheorem{corollary}{Corollary}

\newtheorem{theorem}{Theorem}
\newtheorem{remark}{Remark}
\def\eqn {\begin{equation}}
\def\eeqn {\end{equation}}

\input epsf
\begin{document}
\title{Mathematical Theory of Exchange-Driven Growth}
\date{October 19, 2017}
\author{Emre Esenturk}
\address{University of Warwick, Mathematics Institute, UK}
\email{E.esenturk.1@warwick.ac.uk}
\thanks{Corresponding author email: E.esenturk.1@warwick.ac.uk}

\begin{abstract}
Exchange-driven growth is a process in which pairs of clusters interact and
exchange a single unit of mass. The rate of exchange is given by an
interaction kernel $K(j,k)$ which depends on the masses of the two interacting
clusters. In this paper we establish the fundamental mathematical properties
of the mean field kinetic equations of this process for the first time. We
find two different classes of behaviour depending on whether $K(j,k)$ is
symmetric or not. For the non-symmetric case, we prove global existence and
uniqueness of solutions for kernels satisfying $K(j,k)\leq Cjk$. This result
is optimal in the sense that we show for a large class of initial conditions
with kernels satisfying $K(j,k)\geq Cj^{\beta}$ ($\beta>1)$ the solutions
cannot exist. On the other hand, for symmetric kernels, we prove global
existence of solutions for $K(j,k)\leq C(j^{\mu}k^{\nu}+j^{\nu}k^{\mu})$
($\mu,\nu\leq2,$ $\mu+\nu\leq3),$ while existence is lost for $K(j,k)\geq
Cj^{\beta}$ ($\beta>2).$ In the intermediate regime $3<\mu+\nu\leq4,$ we can
only show local existence. We conjecture that the intermediate regime exhibits
finite-time gelation in accordance with the heuristic results obtained for
particular kernels.

\end{abstract}
\maketitle

\section{Introduction}

Growth processes are ubiquitous in nature. Surprisingly diverse phenomena at
contrasting scales (from microscopic level polymerization processes to cloud
formation to galaxy formation mechanisms at huge scales) have similar driving
mechanisms \cite{Drake}, \cite{Colm-rev}, \cite{Naim-book}. One of the
commonly occuring mechanisms is the cluster growth by coagulation for which
Smoluchoswki and Becker-Doring models are classical examples. For these
models, an extensive mathematical theory has been established \cite{Leyvraz},
\cite{Ball} relating the properties of the cluster size distribution to the
structure of the interaction kernel, $K(j,k)$, encoding the rate of
coagulation of clusters of sizes $j$ and $k$. Exchange-driven growth (EDG) is
another model for non-equilibrium cluster growth which is much less studied.
In EDG pairs of clusters interact by exchanging a single unit of mass
(monomer) \cite{Naim}. In the recent years EDG has also been considered as a
model of social phenomena like migration \cite{Ke2}, population dynamics
\cite{Leyvraz2} and wealth exchange \cite{Isp}. Approaches with similar
spirits found applications in other branches of social sciences \cite{Albi}.
However, no rigorous mathematical results on the EDG type mean-field rate
equations have been obtained to date. So, it is vitally important to do a
rigorous anlaysis of EDG type systems which is the goal of this article.

We note at the outset that, in this article, the time dependent description of
EDG is at the mesoscopic level and we only study the mean field rate equations
(EDG equations) ignoring fluctuations at the particle level. The purpose of
this paper is to provide the mathematical theory on the existence, uniqueness
and non-existence properties of solutions of the EDG equations. It is worth
mentioning that there has recently been increased mathematical interest in the
mass exchange systems since the corresponding kinetic equations (EDG
equations) can be obtained as scaling limits of a class of interacting
particle systems, including zero-range processes \cite{Godrec}, \cite{Stef3},
\cite{Godrec2}, \cite{Beltran}, and more general misanthrope processes
\cite{Stef4}, \cite{Waclaw}, \cite{Colm}, that have been intensively studied
for a range of condensation phenomena that they exhibit Also very recently, it
has been shown that EDG equations can be obtained as limits of a class of
interacting particle systems \cite{Stef3}.

The main mathematical object in our version of the kinetic formulation of the
EDG model is $c_{j}(t)$, the cluster size distirbution, describing the volume
fraction of the system which is occupied by clusters of size $j\geq1.$It is
intuitively clear that, for the classical EDG model which is based on particle
exchange between clusters of non-zero mass, the total mass of the physical
system is conserved. In this study, we consider a modified formulation where
$j=0$ corresponds to the empty (available) volume fraction not occupied by
clusters. As we show later, inclusion of empty volume introduces another
conserved quantity in addition to total mass, and the $c_{j}(t)$ sum to a
constant (or to 1 when normalized with rescaled time) for all times $t>0$.
This formulation is motivated by studies on coarsening dynamics in condensing
particle systems. We note that, the interpretation of the EDG problem
including empty volume or clusters of 'size' 0 is based on a different
motivation than the approach of physicists which does not include volume. The
two approaches are related and our results directly translate to this
classical interpretation, as we will discuss in detail in the conclusion.

Symbolically, the exchange process can be described in the following way. If
$<j>,<k>$ denote the non-zero clusters of sizes $j,k>0,$ then the rule of
interaction is
\[
<j>\oplus<k>\rightarrow<j\pm1>\oplus<k\mp1>.
\]
If, one of the clusters is a zero-cluster ($0$-cluster)$,$ then the rule is
given by
\[
<j>\oplus<0>\rightarrow<j-1>\oplus<1>.
\]

If all the clusters interact uniformly, $K(j,k)c_{j}c_{k}$ denotes the rate of
any cluster of size "$j"$ exporting a single particle to a cluster of size
"$k"$. The details of such microscopic processes are coded in the function
$K(j,k),$ known as the interaction kernel. Depending on the physical or social
system under study the form of the kernel changes. A known physical example of
cluster growth driven by exchange mechanism is the infinite range Kawasaki
\cite{Kawa} zerotemperature spin exchange systems. In this model spin domains
couple by pairwise interaction of perimeter spins, therefore kernel has the
form $K(j,k)=(jk)^{\lambda}$ with $\lambda=d-1,$ $d$ being the dimension (the
exchange rate is proportional to product of number of surface spins). In the
case of social behavior the form of the kernel can be obtained by the culture
or customs of the society \cite{Schel}. For instance, in a (unrealistic)
`non-greedy' society, trades (wealth exchange) would not depend on the
capital, hence the kernel can be assumed to be constant \cite{Isp}.

In most natural occuring systems the rate of these reactions are equal, and it
is common to take $K$ as a symmetric function of its arguments. However, there
are also many processes where export and import of particles do not take place
symmetrically and hence $K(j,k)\neq K(k,j)$ in general. Mathematically, these
generally non-symmetric coupled exchange reactions can be represented by an
infinite set of nonlinear ordinary differential equations (ODEs) with given
initial conditions as below%
\begin{equation}
\dot{c}_{0}=c_{1}\sum_{k=0}^{\infty}K(1,k)c_{k}-c_{0}\sum_{k=1}^{\infty
}K(k,0)c_{k}\text{,} \label{0-infode}%
\end{equation}

\begin{align}
\text{ }\dot{c}_{j}  &  =c_{j+1}\sum_{k=0}^{\infty}K(j+1,k)c_{k}-c_{j}%
\sum_{k=0}^{\infty}K(j,k)c_{k}\label{infode}\\
&  -c_{j}\sum_{k=1}^{\infty}K(k,j)c_{k}+c_{j-1}\sum_{k=1}^{\infty
}K(k,j-1)c_{k}\text{ ,\ }%
\end{align}%
\begin{equation}
c_{j}(0)=c_{j,0}\text{ \ \ \ }\{j=0,1,2,...\}. \label{infIC}%
\end{equation}

\bigskip In this article our main goal is to prove the fundamental properties
of this infinite system of equations such as existence of global solutions,
uniqueness, positivity and possible cases leading to non-existence. In order
to put our work into context, we give a brief summary of other growth systems
which have been extensively studied.

Basic aggregation models are quite old and date back to the works of
Smoluchowski \cite{Smo} (1917) and Becker-Doring \cite{Becker} (1935) (see
\cite{Leyvraz} for other related works). Over the decades, systematic
mathematical analysis of the resulting equations have been carried out
\cite{McLeod}, \cite{Ball} and mathematical questions concerning existence and
uniqueness of these systems have been investigated in fair generality for
kernels satisfying bounds, $K(j,k)\leq C$ \cite{Melzak}$,$ $K(j,k)\leq C(j+k)$
\cite{White}$,$ $K(j,k)\leq Ca(j)a(k)$ ($a(j)=o(j))$ \cite{Leyvraz3}.

One of the striking results of these studies was that when the interaction
kernel grows fast enough, drastic changes take place in the dynamics of the
problem. For instance, when the kernel is super-linear the solutions ceases to
exist \cite{Ball} for the Becker-Doring model, while in the Smoluchowski
model, the system undergoes a phase transition and begins behaving very
differently. The latter case, known as gelation \cite{Ziff}, \cite{Esco},
\cite{Menon} is a counter-intuitive phenomenon where some of the mass in the
system "escapes" to infinity. At the same time the uniqueness of the solution
is lost along with a change in scaling behavior. So, it is physically and
mathematically very important to identify the regions where such strange
behaviors may happen.

For the exchange-driven growth problem, heuristic studies suggest \cite{Naim}
that for symmetric kernels of the form $K(j,k)=(jk)^{\mu},$ no gelation occurs
if $\mu\leq3/2$ (regular case). When $2\geq\mu>3/2$ however, gelation takes
place at some finite time $T_{g}$. For, $\mu>2,$ even more strangely, gelation
takes place right at the beginning at $t=0,$ known as instantaneous gelation.
This behavior is significantly different from the Smoluchowski model in which
ordinary gelation occurs for $1\geq\mu>1/2$ and post gel solutions continue to
exist for $t>T_{g},$ while instantaneous gelation takes place for $\mu>1$
\cite{Dongen}, \cite{Carr}$.$

In this article, we investigate the both regular and singular cases for the
EDG problem in the sense described above. In particular, we prove rigorously,
for a system with general non-symmetric kernel satisfying the bound
$K(j,k)\leq Cjk$ that the solution exists globally and is unique and conserves
the mass. However, if the growth of the kernel is faster, i.e., $K(j,k)\geq
Cj^{\beta}$ ($\beta>1)$ then under some assumptions on the initial conditions,
the solutions can be shown to be non-existent. So, in this sense the growth
rate on the kernel for global existence is optimal. For symmetric kernels, the
results can be extended considerably. We prove that, if $K(j,k)\leq C(j^{\mu
}k^{\nu}+$ $j^{\nu}k^{\mu})$ ($\mu,\nu\leq2,$ $\mu+\nu\leq3)$ then the
solutions are global and mass-conserving. We also identify an intermediate
regime ($\mu,\nu\leq2,$ $\mu+\nu\leq4)$ where the solutions exist locally. We
conjecture that this is the gelation regime where there is a loss of mass
after a finite time (the gelation time). Beyond this regime, i.e., if
$K(j,k)\geq Cj^{\beta}$ ($\beta>2)$ once gain we show that the solutions cease
to exist.

To prove the existence we employ a truncation method (due to McLeod)
\cite{McLeod}, \cite{McLeod2} which suits well to the discrete structure of
the equations. The truncated finite ODE system is useful in providing basic
estimates on the total mass allowing one to pass to the limit which we will
prove to solve the original (infinite) ODE system. The main assumption is that
initial cluster distributions decay sufficiently fast (some higher moments
exist). For the symmetric kernels, we show that one can actually obtain better
estimates than just bounding the total mass (which is intuitively obvious).
The arguments follow by fortunate cancellations due to symmetry and use of
some fundamental inequalities. For the uniqueness of solutions we provide two
results for the non-symmetric and symmetric kernels. The ideas are based on
controlling the difference of (supposedly distinct) solutions. Again, one
needs to produce different (but similar) routes of steps for the two cases
(non-symmetric and symmetric kernels). The non-existence, on the other hand,
is based on the idea of obtaining lower bounds to the tails of the
distributions and arguing that these lead to contradictions. To prove the
non-existence for the non-symmetric kernel we need to make additional
assumption that the kernel selectively favors growth. For the symmetric
kernel, we do not need such selectivity (and it is clearly disallowed by the
symmetry). However, in that case, non-existence will take place only for fast
growing kernels (faster than quadratic) as expected.

The structure of this article is as follows. In Section 2, we detail the
truncation method and show some of its basic properties which hold true
uniformly for arbitrarily large finite systems. We then use these preliminary
results to prove, after a number of technical steps, global existence of
solutions for the non-symmetric and symmetric kernels. In Section 3, we show
the other important results related to the same EDG system: uniqueness,
positivity and non-existence of solutions. In Section 4, we conclude the paper
by discussing the relationship between our formulation of the problem and
existing physics literature. We also point out possible extensions of the
current work and suggest some other future research directions.

\section{Existence of Solutions}

We start by giving the setting of the problem and some definitions. Let
$X_{\mu}=\{x=(x_{j}),$ $x_{j}\in\mathbb{R};\left\Vert x\right\Vert _{\mu
}<\infty\}$ be the space of sequences equipped with the norm $\left\Vert
x\right\Vert _{\mu}=\sum_{j=1}^{\infty}j^{\mu}x_{j}$ where $\mu\geq0.$ Also,
let $K(\cdot,\cdot):\mathbb{R}\times\mathbb{R}\rightarrow\lbrack0,\infty)$ be
the cluster interaction kernel which we assume to be non-negative
throughout$.$ We set $K(0,j)\equiv0$ identically.

\textbf{Definition 1: }We say the system\ has a solutions iff

$(i)$ $c_{j}(t)$ $:[0,\infty)\rightarrow\lbrack0,\infty)$ is continuous and
$\sup_{t\in\lbrack0,\infty)}c_{j}(t)<\infty$

$(ii)$ $\int_{0}^{t}\sum_{k=0}^{\infty}K(j,k)c_{k}ds<\infty,$ $\int_{0}%
^{t}\sum_{k=1}^{\infty}K(k,j)c_{k}ds<\infty$ for all $j\in\mathbb{N}$ and
$t\in\lbrack0,T)$ ($T\leq\infty)$

$(iii)\ c_{j}(t)=c_{j}(0)+\int_{0}^{t}\left(  c_{j+1}\sum_{k=0}^{\infty
}K(j+1,k)c_{k}-c_{j}\sum_{k=0}^{\infty}K(j,k)c_{k}\right)  ds$

$\ \ \ \ \ \ \ \ \ \ \ \ \ \ \ \ \ \ \ \ \ \ \ \ \ \ \ \ +\int_{0}^{t}\left(
-c_{j}\sum_{k=1}^{\infty}K(k,j)c_{k}+c_{j-1}\sum_{k=1}^{\infty}K(k,j-1)c_{k}%
\right)  ds$ \ $\{j>0\}$

\ \ \ \ \ \ $c_{0}(t)=c_{0}(0)+\int_{0}^{t}c_{1}\sum_{k=0}^{\infty}%
K(1,k)c_{k}-c_{0}\sum_{k=1}^{\infty}K(k,0)c_{k}.$

\textbf{Definition 2:} For a sequence $(c_{j})_{j=1}^{N}$, we call the
quantity $M_{p}^{N}(t)=\sum_{j=0}^{N}j^{p}c_{j}(t)$ as the $p^{th}-$moment of
the sequence. If the sequence is infinite, then we denote the $p^{th}-$moment
with $M_{p}(t)=\sum_{j=0}^{\infty}j^{p}c_{j}(t).$

\textbf{Definition 3:} We say that the kernel $K(j,k)$ is nearly symmetric iff
$K(j,k)=K(k,j)$ for all $j,k\geq1.$

To prove the existence, we first consider a truncated system which respects,
even at the finite dimensional level, the key features of the original
infinite dimensional ODE system. Then, we obtain, for the truncated system,
some uniform bounds. With the help of these bounds the limit of the truncated
system is shown to be well defined and is actually a solution of the original problem.

Now, consider the truncated EDG system where we cut off the equations at a
finite order $N$ (that is, setting $c_{j}\equiv0$ identically for $j>N)$
\begin{equation}
\text{\ }\dot{c}_{0}^{N}=c_{1}^{N}\sum_{k=0}^{N-1}K(1,k)c_{k}^{N}-c_{0}%
^{N}\sum_{k=1}^{N}K(k,0)c_{k}^{N}, \label{Tode0}%
\end{equation}

\begin{align}
\text{\ }\dot{c}_{j}^{N}  &  =c_{j+1}^{N}\sum_{k=0}^{N-1}K(j+1,k)c_{k}%
^{N}-c_{j}^{N}\sum_{k=0}^{N-1}K(j,k)c_{k}^{N}\\
&  -c_{j}^{N}\sum_{k=1}^{N}K(k,j)c_{k}^{N}+c_{j-1}^{N}\sum_{k=1}%
^{N}K(k,j-1)c_{k}^{N},\text{ }\{1\leq j\leq N-1\}\text{\ \ }\nonumber
\end{align}%
\begin{equation}
\dot{c}_{N}^{N}=-c_{N}^{N}\sum_{k=0}^{N-1}K(N,k)c_{k}^{N}+c_{N-1}^{N}%
\sum_{k=1}^{N}K(k,N-1)c_{k}^{N}, \label{TodeN}%
\end{equation}
with the initial conditions given by%
\begin{equation}
c_{j}^{N}(0)=c_{j,0}\geq0,\text{ \ }\{0\leq j\leq N\}. \label{TodeIC}%
\end{equation}
The existence and uniqueness of this system comes from the standard ODE
theory. It is also known that the solutions are continuously differentiable.

Next, some preliminary lemmas are in order. The first lemma below demonstrates
(as a corollary) that the truncated system has two conserved quantities. The
significance of this result will shortly be clear when getting the uniform
estimates (in $N$) for the growth of cluster size distributions.

\begin{lemma}
\bigskip Let $g_{j}$ be a sequence of non-negative real numbers. Then,
\begin{equation}
\sum_{j=0}^{N}g_{j}\frac{dc_{j}^{N}}{dt}=\sum_{j=1}^{N}(g_{j-1}-g_{j}%
)c_{j}^{N}\sum_{k=0}^{N-1}K(j,k)c_{k}^{N}+\sum_{j=0}^{N-1}(-g_{j}%
+g_{j+1})c_{j}^{N}\sum_{k=1}^{N}K(k,j)c_{k}^{N}. \label{mom-red1}%
\end{equation}
If $K(\cdot,\cdot)$ is nearly symmetric, then one has
\begin{align}
\sum_{j=0}^{N}g_{j}\frac{dc_{j}^{N}}{dt}  &  =\sum_{j=1}^{N-1}(g_{j-1}%
-2g_{j}+g_{j-1})c_{j}^{N}\sum_{k=1}^{N-1}K(j,k)c_{k}^{N}\label{mom-red2}\\
&  +\sum_{j=1}^{N-1}\left(  (g_{j-1}-g_{j})+(g_{1}-g_{0})\right)
K(j,0)c_{j}^{N}c_{0}^{N}\nonumber\\
&  +\sum_{j=1}^{N-1}((g_{j+1}-g_{j})+(g_{N-1}-g_{N}))c_{j}^{N}K(N,j)c_{N}%
^{N}\nonumber\\
&  +((g_{N-1}-g_{N})+(g_{1}-g_{0}))c_{j}^{N}K(N,0)c_{0}^{N}.
\end{align}

\end{lemma}

\begin{proof}
Writing $\dot{c}_{j}^{N}(t)$ from (\ref{Tode0})-(\ref{TodeN}) and taking the
summation for the $g(j)\dot{c}_{j}^{N}$ and shifting the indices on the terms
having $c_{j+1},c_{j-1}$, we get%
\begin{align}
\sum_{j=0}^{N}g_{j}\frac{dc_{j}^{N}}{dt}  &  =\sum_{j=1}^{N}g_{j-1}c_{j}%
^{N}\sum_{k=0}^{N-1}K(j,k)c_{k}^{N}-\sum_{j=1}^{N}g_{j}c_{j}^{N}\sum
_{k=0}^{N-1}K(j,k)c_{k}^{N}\label{g-form1}\\
&  -\sum_{j=0}^{N-1}g_{j}c_{j}^{N}\sum_{k=1}^{N}K(k,j)c_{k}^{N}+\sum
_{j=0}^{N-1}g_{j+1}c_{j}^{N}\sum_{k=1}^{N}K(k,j)c_{k}^{N}. \label{g-form2}%
\end{align}
Collecting the $1^{st}$,$2^{nd}$ and $3^{rd},4^{th}$ terms in (\ref{g-form1}%
),(\ref{g-form2}) together yields the first identity
\begin{align}
\sum_{j=0}^{N}g_{j}\frac{dc_{j}^{N}}{dt}  &  =\sum_{j=1}^{N}(g_{j-1}%
-g_{j})c_{j}^{N}\sum_{k=0}^{N-1}K(j,k)c_{k}^{N}\label{gj-sym1}\\
&  +\sum_{j=0}^{N-1}(g_{j+1}-g_{j})c_{j}^{N}\sum_{k=1}^{N}K(k,j)c_{k}^{N}.
\label{gj-sym2}%
\end{align}

For the second identity we first split the sums in (\ref{gj-sym1}),
(\ref{gj-sym2}) and recombine the terms that are alike, while accounting for
the "boundary terms"$.$ Let $A,$ $B~$denote the sums on the right hand side of
(\ref{gj-sym1}) and (\ref{gj-sym2}). Then, one has%
\begin{align*}
A  &  =\sum_{j=1}^{N-1}(g_{j-1}-g_{j})c_{j}^{N}\sum_{k=1}^{N-1}K(j,k)c_{k}%
^{N}+\sum_{j=1}^{N-1}(g_{j-1}-g_{j})c_{j}K(j,0)c_{0}^{N}\\
&  +(g_{N-1}-g_{N})c_{N}^{N}\sum_{k=1}^{N-1}K(N,k)c_{k}^{N}+(g_{N-1}%
-g_{N})c_{N}^{N}K(N,0)c_{0}^{N},
\end{align*}%
\begin{align*}
B  &  =(g_{1}-g_{0})c_{0}^{N}\sum_{k=1}^{N-1}K(k,0)c_{k}^{N}+(g_{1}%
-g_{0})c_{j}^{N}K(N,0)c_{0}^{N}\\
&  +\sum_{j=1}^{N-1}(g_{j+1}-g_{j})c_{j}^{N}\sum_{k=1}^{N-1}K(k,j)c_{k}%
^{N}+\sum_{j=1}^{N-1}(g_{j+1}-g_{j})c_{j}^{N}K(N,j)c_{N}^{N}.
\end{align*}
Taking the sum $A+B,$ rearranging the terms and using the symmetry of $K$
yields result.
\end{proof}

\begin{corollary}
For a general kernel $K,$ the zeroth moment and the first moment of the
truncated system (\ref{Tode0})-(\ref{TodeIC}) are conserved in time.
\end{corollary}

\begin{proof}
By setting $g_{j}=1,$ we see that all the terms in the first identity of Lemma
1 cancels each other%
\[
\sum_{j=0}^{N}\dot{c}_{j}^{N}(t)=0,
\]
and hence the zeroth moment is conserved. To see that the first moment is also
conserved we set $g_{j}=j.$ Then again, by the first identity of Lemma $1$ we
get%
\[
\sum_{j=0}^{N}\dot{c}_{j}^{N}(t)=\sum_{j=1}^{N}(-1)c_{j}^{N}\sum_{k=0}%
^{N-1}K(j,k)c_{k}^{N}+\sum_{j=0}^{N-1}(1)c_{j}^{N}\sum_{k=1}^{N}%
K(k,j)c_{k}^{N}=0,
\]
which gives conservation of the first moment.
\end{proof}

For the proofs of existence theorems, we will also need the following lemma
which shows the non-negativity of solutions of the truncated system if the
initial cluster distributions are non-negative.

\begin{lemma}
\bigskip Let $c_{j}^{N}(t)$ be a solution of the truncated system
(\ref{Tode0})-(\ref{TodeIC}) where $K(j,k)\geq0$. If $c_{j}^{N}(0)\geq0$ for
all $j\geq0,$ then $c_{j}^{N}(t)\geq0.$
\end{lemma}

\begin{proof}
Let $S(j,c^{N})=\sum_{k=0}^{N-1}K(j,k)c_{k}^{N}$ and$~\bar{S}(j,c^{N}%
)=\sum_{k=1}^{N}K(k,j)c_{k}^{N}.$ Then the system (\ref{Tode0})-(\ref{TodeIC})
can be written as%
\[
\frac{dc_{0}^{N}}{dt}+\bar{S}(0,c^{N})c_{0}^{N}=S(1,c^{N})c_{1}^{N},
\]%
\begin{equation}
\frac{dc_{j}^{N}}{dt}+(S(j,c^{N})+\bar{S}(j,c^{N}))c_{j}=c_{j+1}%
^{N}(t)S(j+1,c^{N})+c_{j-1}^{N}(t)\bar{S}(j-1,c^{N})\text{ \ }\{N>j\geq1\}.
\label{fin-pos}%
\end{equation}%
\[
\frac{dc_{N}^{N}}{dt}+c_{N}^{N}S(N,c^{N})=c_{N-1}^{N}\bar{S}(N-1,c^{N})
\]

Now, if the assertion in the theorem were not true, then there would be a very
first time $t_{0}\in\lbrack0,\tau)$ and some $i\in\mathbb{N}$, such that
$c_{i}^{N}(t_{0})=0$ and $(c_{i}^{N})^{\prime}(t_{0})<0.$ Suppose $i>0$
(similar argument can be repeated if $i=0).$ Then for the left hand side of
(\ref{fin-pos}) we have
\begin{equation}
\frac{dc_{i}^{N}(t_{0})}{dt}+(S(i,c^{N}(t_{0}))+\bar{S}(i,c^{N}(t_{0}%
)))c_{i}^{N}(t_{0})<0. \label{fin-neg}%
\end{equation}
However, the right hand side of (\ref{fin-pos}) gives
\begin{equation}
c_{i+1}^{N}(t_{0})S(i+1,c^{N}(t_{0}))+c_{i-1}^{N}(t_{0})\bar{S}(i-1,c^{N}%
(t_{0}))\geq0
\end{equation}
since $c_{j}(t_{0})\geq0$. But this contradicts with (\ref{fin-neg}). Hence we
have $c_{j}(t)\geq0$ for all $j$ and $t.$
\end{proof}

Now, we state and prove the main theorems of this section. We provide two
different versions of the existence theorems for each of the non-symmetric and
nearly symmetric kernel cases. As the assumptions of the theorems are
different, the results do not imply each other. In the first version, we
demand more on the moments of the initial cluster distribution. This was the
approach taken in \cite{White} for the Smoluchowski equation. In the second
version we demand more on the growth of the kernel.

In the sequel, we denote, by $C\geq0,$ a dummy constant which may take
different values at different steps.

\begin{theorem}
Consider the EDG system given by (\ref{0-infode})-(\ref{infIC}). Let $K(j,k)$
be a general kernel satisfying $K(j,k)\leq Cjk$ for large enough $j,k.$ Assume
further that $M_{p}(0)=\sum_{k=0}^{\infty}j^{p}c_{j}(0)<\infty$ for some
$p>1$. Then the infinite system (\ref{0-infode})-(\ref{infIC}) has a global
solution $(c_{j})\in X_{1}$.
\end{theorem}

\begin{proof}
The key ingredient of the proof is the constancy of the zeroth and first
moment of the truncated system $M_{1}^{N}(t)$. This then will imply that
$c_{j}^{N}(t)$ and $\dot{c}_{j}^{N}(t)$ and are bounded uniformly. Indeed,
since $c_{j}^{N}(t)$ are non-negative, the bound on the zeroth moment
\[
\sum_{j=0}^{N}c_{j}^{N}(t)=\sum_{j=0}^{N}c_{j}^{N}(0)\leq\sum_{j=0}^{\infty
}c_{j}(0)=M_{0}(0)
\]
yields $c_{j}^{N}(t)\leq M_{0}(0)$ for all $N$ and $j\geq0.$ Similarly, for
the derivatives, we have (when $j\geq1)$
\begin{align*}
\left\vert \dot{c}_{j}^{N}(t)\right\vert  &  \leq\sum_{k=0}^{N-1}c_{j+1}%
^{N}K(j+1,k)c_{k}^{N}+c_{j}^{N}\sum_{k=0}^{N-1}K(j,k)c_{k}^{N}\\
&  +\sum_{k=1}^{N}K(k,j)c_{k}^{N}c_{j}^{N}+\sum_{k=1}^{N}K(k,j-1)c_{k}%
^{N}c_{j-1}^{N}\\
&  \leq C\sum_{k=0}^{N}jkc_{j}^{N}c_{k}^{N}\leq CM_{1}(0)^{2}.
\end{align*}
where, to get to the third line, we simply shifted the "$j"$ indices and used
the bound on $K(j,k).$ Similarly we can show $\left\vert \dot{c}_{0}%
^{N}(t)\right\vert \leq C$. Hence the sequence $(c_{j}^{N})$ is uniformly
bounded and equicontinuous. Then by Arzela-Ascoli theorem there is a
subsequence $\{c_{j}^{N(i)}\}$ which converges uniformly to a continuous
function, say $c_{j}(t).$ Let us denote the subsequence $N(i)$ also with $N$
for brevity. To show that $c_{j}(t)$ is a solution to the original problem we
need to show the series $\sum_{j=1}^{N}K(j,k)c_{k}^{N}$ converges uniformly on
bounded intervals of time $[0,T].$ To prove this, we need the boundedness of a
higher moment. Let $g(s)=s^{p}$ for some $1<p\leq2$ without loss of
generality$.$ By the mean value theorem $j^{p}-(j-1)^{p}=p(j-\theta_{1}%
)^{p-1}$ and $(j+1)^{p}-j^{p}=p(j+\theta_{2})^{p-1}$ for some $0<\theta
_{1},\theta_{2}<1.$ Then, from the first identity in Lemma 1%
\begin{align*}
\dot{M}_{p}^{N}(t)  &  =\sum_{j=0}^{N}j^{p}\dot{c}_{j}(t)=\sum_{j=1}%
^{N}p(j-\theta_{1})^{p-1}c_{j}^{N}\sum_{k=0}^{N-1}K(j,k)c_{k}^{N}+\sum
_{j=0}^{N-1}p(j+\theta_{2})^{p-1}c_{j}^{N}\sum_{k=1}^{N}K(k,j)c_{k}^{N}\\
&  \leq\sum_{j=1}^{N}pj^{p-1}jc_{j}^{N}\sum_{k=0}^{N-1}kc_{k}^{N}+\sum
_{j=1}^{N-1}pj(j+1)^{p-1}c_{j}^{N}\sum_{k=1}^{N}kc_{k}^{N}\leq CM_{p}%
^{N}(t)M_{1}(0).
\end{align*}
Hence one has $M_{p}^{N}(t)\leq M_{p}^{N}(0)e^{Ct}\leq M_{p}(0)e^{Ct}$ by
Gronwall inequality. Now, $\sum_{j=1}^{N-1}K(j,k)c_{k}^{N}$ converges
uniformly to $\sum_{j=1}^{\infty}K(j,k)c_{k}$. To see this we observe%
\begin{equation}
\left\vert \sum_{k=1}^{\infty}K(j,k)c_{k}^{N}-\sum_{k=1}^{\infty}%
K(j,k)c_{k}\right\vert \leq\sum_{k=1}^{N_{2}}K(j,k)\left\vert c_{k}^{N}%
-c_{k}\right\vert +\left\vert \sum_{k=N_{2}+1}^{\infty}K(j,k)(c_{k}+c_{k}%
^{N})\right\vert . \label{unidif1}%
\end{equation}
In the limit, the second term on the right hand side of (\ref{unidif1}) can be
made arbitrarily small for $N_{2}$ large enough since%
\[
\left\vert \sum_{k=N_{2}+1}^{\infty}K(j,k)(c_{k}+c_{k}^{N})\right\vert
\leq2Cj\sum_{k=N_{2}+1}^{\infty}kk^{-p}k^{p}(c_{k}+c_{k}^{N})\leq
CjN_{2}^{1-p}M_{p}^{N}(t).
\]
The first term on the right hand side of (\ref{unidif1}) can be made
arbitrarily small be letting $N$ become large. Hence $\sum_{k=1}^{\infty
}K(j,k)c_{k}^{N}$ converges uniformly. Similarly, $\sum_{k=1}^{N}%
K(k,j)c_{k}^{N}$ also converges uniformly. Now, if we write the truncated
system in the integral form%
\begin{align}
\text{\ }c_{j}^{N}(t)  &  =c_{j}^{N}(0)+\int_{0}^{t}c_{j+1}^{N}(s)\sum
_{k=0}^{N-1}K(j+1,k)c_{k}^{N}(s)-\int_{0}^{t}c_{j}^{N}\sum_{k=0}%
^{N-1}K(j,k)c_{k}^{N}(s)ds\\
&  -\int_{0}^{t}c_{j}^{N}(s)\sum_{k=1}^{N}K(k,j)c_{k}^{N}(s)+\int_{0}%
^{t}c_{j-1}^{N}(s)\sum_{k=1}^{N}K(k,j-1)c_{k}^{N}(s)ds\text{ \ \ }\nonumber
\end{align}
we see that we can pass to the limit $N\rightarrow\infty,$ on the right hand
side, under the integral sign since the functions $c_{j}^{N}(t)$\ and
$\sum_{k=1}^{N-1}K(j,k)c_{k}^{N}$ converge uniformly. This shows that $c_{j}$
as the limit, is a solution of the system (\ref{infode})-(\ref{infIC}).
\end{proof}

\bigskip From the construction in the above theorem, considering the integral
form of the equations, it is immediate that the limit solution $c_{j}(t)$ is
differentiable due to the uniform convergence of $c_{j}^{N}$ and the sums
involved. We also note that, under the conditions of Theorem 1, with the
boundedness of the higher moments, i.e., $M_{p}^{N}(t)<C(c(0),t)<\infty$ for
$p>1,$ the approximate (truncated) solutions converge strongly to the limit
function, i.e., $\lim_{i\rightarrow\infty}\left\Vert c_{j}^{N(i)}%
(t)-c_{j}(t)\right\Vert _{\mu}\rightarrow0$ for $\mu<p$. In particular, we
have the following corollary as a consequence.

\begin{corollary}
Let $c_{j}$ be the solution of the (\ref{infode})-(\ref{infIC}) under the
conditions of Theorem 1 for some $p>1$. Then $c_{j}(t)$ is continuously
differentiable. Moreover $M_{p}(t)<\infty$ and
\begin{align*}
\sum_{0}^{\infty}c_{j}(t)  &  =\sum_{0}^{\infty}c_{j}(0),\\
\sum_{0}^{\infty}jc_{j}(t)  &  =\sum_{0}^{\infty}jc_{j}(0).
\end{align*}

\end{corollary}

If the kernel $K$ is nearly symmetric, by some further cancellations and use
of a simple inductive argument together with a fundamental inequality, we can
prove a stronger result for exponents satisfying $\mu+\nu\leq3$.

\begin{theorem}
Consider the infinite EDG system (\ref{infode})-(\ref{infIC}). Let $K(j,k)$ be
nearly symmetric and satisfy $K(j,k)\leq C(j^{\mu}k^{\nu}+j^{\nu}k^{\mu})$
$(\mu+\nu\leq3,$ $\mu,\nu\leq2)$ and $M_{p}(0)<\infty$ for some $p>2.$ Then
the system (\ref{infode})-(\ref{infIC}) has a global solution $(c_{j})\in
X_{2}$.
\end{theorem}

\begin{proof}
The general idea of the proof is similar to the previous one. However, we now
allow faster growth on $K$ and therefore, boundedness of $M_{1}(t)$ is not
sufficient. We need estimates on the higher moments which will be done by
bounding uniformly the higher moments of the truncated system. To see this, we
use the second identity in the Lemma 1.

Let us first show that the second moment of the truncated system is uniformly
bounded$.$ We first observe that, in Lemma 1, the second line of
(\ref{mom-red2}) is non-positive. Indeed, choosing $g_{j}=j^{2}$ we have, for
$1\leq j\leq N-1,$%
\[
(g_{j-1}-g_{j})+(g_{1}-g_{0})=-2j+2\leq0
\]
Similarly, the second and third lines are also non-positive since $j\leq N-1,$
giving%
\[
(g_{j+1}-g_{j})+(g_{N-1}-g_{N})=2j+1-2N+1\leq0,
\]%
\[
(g_{1}-g_{0})-(g_{N}-g_{N-1})=1-(2N-1)\leq0.
\]
Then we have the following inequality for $M_{2}^{N}(t)$
\begin{align}
\sum_{0}^{N}j^{2}\dot{c}_{j}^{N}(t)  &  \leq\sum_{j=0}^{N-1}((j+1)^{2}%
-2j^{2}+(j-1)^{2})c_{j}^{N}\sum_{k=0}^{N-1}K(j,k)c_{k}^{N}\\
&  \leq2C\sum_{j=0}^{N-1}\sum_{k=0}^{N-1}(j^{\mu}k^{\nu}+j^{\nu}k^{\mu}%
)c_{j}^{N}c_{k}^{N}. \label{M2-sym}%
\end{align}
Now since $\mu+\nu\leq3$ for the exponents in (\ref{M2-sym}), there exists
$\bar{\mu}\geq\mu$ and $\bar{\nu}\geq\nu$ such that $\bar{\mu}+\bar{\nu}=3$.
Now, by Young's inequality we have
\begin{equation}
j^{\bar{\mu}}k^{\bar{\nu}}\leq\left(  \frac{2\bar{\mu}-\bar{\nu}}{3}%
j^{2}k+\frac{2\bar{\nu}-\bar{\mu}}{3}jk^{2}\right)  . \label{young}%
\end{equation}
Then (\ref{young}) and inequality (\ref{M2-sym}) together give
\begin{align*}
\dot{M}_{2}^{N}(t)  &  \leq C\sum_{j=0}^{N-1}\sum_{k=0}^{N-1}(j^{\bar{\mu}%
}k^{\bar{\nu}}+j^{\bar{\nu}}k^{\bar{\mu}})c_{j}^{N}c_{k}^{N}\leq C\sum
_{j=0}^{N-1}\sum_{k=0}^{N-1}(j^{2}k+jk^{2})c_{j}^{N}c_{k}^{N}\\
&  \leq CM_{2}^{N}(t)M_{1}^{N}(t)\leq CM_{2}^{N}(t),
\end{align*}
from which we deduce, by Gronwall's inequality,%
\[
M_{2}^{N}(t)\leq M_{2}(0)e^{Ct},
\]
which is a uniform bound for all $N.$ Then arguing as in Theorem 1 we find a
subsequence $c_{j}^{N(i)}(t)$ which converges uniformly to $c_{j}(t)$.
However, to prove that $c_{j}$ is a solution in the sense of Definition 1 we
need boundedness of higher moments, i.e., $M_{p}(t)<\infty$ for some $p>2$.
But, this now can be achieved using the boundedness of $M_{2}(t)$ which we
just have proved$.$ Indeed, let $g_{j}=j^{p}$ and take, without loss of
generality, $2<p<3$ where $M_{p}(0)<\infty.$ Then, by the mean value theorem
we see%
\[
(g_{j+1}-g_{j})-(g_{N}-g_{N-1})=(j+\theta_{1})^{p-1}-(N-1+\theta_{2}%
)^{p-1}\leq0.
\]
Similarly $(g_{j-1}-g_{j})+(g_{1}-g_{0})\leq0$ and $(g_{1}-g_{0}%
)-(g_{N}-g_{N-1})\leq0$, Hence, by Lemma 1, we have
\[
\sum_{0}^{N}j^{p}\dot{c}_{j}^{N}(t)\leq\sum_{j=0}^{N-1}((j+1)^{p}%
-2j^{p}+(j-1)^{p})c_{j}^{N}\sum_{k=0}^{N-1}K(j,k)c_{k}^{N}.
\]
Expanding the function $g(s)=s^{p}$ around $s=j$ in Taylor series up to second
order gives $\left\vert (j+1)^{p}-2j^{p}+(j-1)^{p}\right\vert \leq Cj^{p-2}$
and hence
\[
\sum_{0}^{N}j^{p}\dot{c}_{j}^{N}(t)\leq C\sum_{j=0}^{N-1}\sum_{k=0}%
^{N-1}j^{p-2}(j^{\mu}k^{\nu}+j^{\nu}k^{\mu})c_{j}^{N}(t)c_{k}^{N}(t)\leq
C\sum_{j=0}^{N-1}\sum_{k=0}^{N-1}j^{p-2}(j^{2}k+jk^{2})c_{j}^{N}(t)c_{k}%
^{N}(t)
\]
In the second step above we again used Young's inequality. Taking the sums on
the furthest right yields
\[
\sum_{0}^{N}j^{p}\dot{c}_{j}^{N}(t)\leq C(M_{p}^{N}(t)M_{1}+M_{p-1}%
^{N}(t)M_{2}(t))\leq C(M_{1}+M_{2}(t))M_{p}^{N}(t)
\]
which, by another use of Gronwall inequality, gives the bound $M_{p}%
^{N}(t)\leq M_{p}(0)e^{\int_{0}^{t}C(M_{1}+M_{2}(s))ds}.$ Repeating the
arguments in Theorem 1 proves that $c_{j}(t)$ is indeed a solution.
\end{proof}

\begin{remark}
\bigskip The growth assumption $K(j,k)\leq C(j^{\mu}k^{\nu}+j^{\nu}k^{\mu})$
$(\mu+\nu\leq3,$ $\mu,\nu\leq2)$ in the theorem was crucial to get the global
existence. This is in accordance with the physical studies which found regular
growth for the same regime assuming specific forms for the kernels. For
general symmetric kernels growing faster than the aforementioned rates we can
only prove local existence of solutions as shown in the following corollary.
\end{remark}

\begin{corollary}
\bigskip Consider the infinite EDG system (\ref{0-infode})-(\ref{infIC}). Let
$K(j,k)$ be nearly symmetric and satisfy $K(j,k)\leq Cj^{2}k^{2}$ (for $j,k$
large) and $M_{p}(0)<\infty$ for some $p>2.$ Then the system (\ref{0-infode}%
)-(\ref{infIC}) has a local solution $(c_{j})\in X_{2}.$
\end{corollary}

\begin{proof}
The proof takes similar steps to Theorem 3. Indeed, under the assumption
$K(j,k)\leq j^{2}k^{2}$, we again consider $\dot{M}_{2}^{N}(t)$
\begin{equation}
\sum_{0}^{N}j^{2}\dot{c}_{j}^{N}(t)\leq2C\sum_{j=0}^{N-1}\sum_{k=0}^{N-1}%
j^{2}k^{2}c_{j}^{N}c_{k}^{N}\leq2C\left(  M_{2}^{N}(t)\right)  ^{2}.
\end{equation}
Hence we can obtain
\[
M_{2}^{N}(t)\leq\frac{1}{\frac{1}{M_{2}^{N}(0)}-2Ct}\leq\frac{1}{\frac
{1}{M_{2}(0)}-2Ct}\text{ for }t<1/(2M_{2}(0)C),
\]
a uniform bound which is valid up to some certain finite time$.$ This
nevertheless allows us to construct a subsequence $c_{j}^{N(i)}$, as before,
which converges uniformly to a limit function $c_{j}(t).$ We can then get a
bound for $M_{p}(t)$ (valid up to a finite time $T$) and show that the partial
sums in the truncated system converges uniformly up to time $T$ which proves
the existence of local solutions.
\end{proof}

Theorem 2 and the corollary that follows give us signs of an intermediate
regime where the solutions behave differently. Previous heuristic studies with
special kernels of the form $K(j,k)=j^{\mu}k^{\nu}+j^{\nu}k^{\mu}$ suggest
that $\mu+\nu=3$ is the critical line for the onset of finite time gelation.
So, in light of the previous theorems we can make the following conjecture.

\textbf{Conjecture:} \textit{Consider the infinite EDG system (\ref{0-infode}%
)-(\ref{infIC}). Let the nearly symmetric kernel satisfy }$K(j,k)\geq Cj^{\mu
}k^{\nu},$\textit{ }$\mu+\nu>3.$\textit{ Then gelation occurs in finite time.}

The previous two theorems crucially made use of the boundedness of the initial
moments. We can relax this assumption by sacrificing on the growth rate of
$K.$ This was the approach taken by \cite{Leyvraz3} for the Smoluchowski
equation. More precisely, if we assume, for the non-symmetric kernel, the
growth rate%
\[
(I):K(j,k)\leq a(j)b(k)\text{ \ with }a(j),b(j)=o(j)
\]
then we have the following.

\begin{theorem}
Consider the EDG system given by (\ref{0-infode})-(\ref{infIC}). Let $K(j,k)$
be a general kernel satisfying the growth condition $(I)$ above. Suppose that
the system has finite initial total mass. Then the infinite system
(\ref{0-infode})-(\ref{infIC}) has a global solution $(c_{j})\in X_{1}$.
\end{theorem}

\begin{proof}
As in Theorem 1, we can use the boundedness of the zeroth and first moments of
the truncated system (for any $N)$ to construct a sequence of solutions that
converge uniformly to a continuous function on bounded time intervals $[0,T].$
However, to show that this is the desired solution, we also need prove that
$\sum_{k=0}^{N-1}K(j,k)c_{k}^{N}\rightrightarrows\sum_{k=0}^{\infty
}K(j,k)c_{k}$ as $N\rightarrow\infty.$ This can be shown using the growth rate
of the kernel%
\[
\left\vert \sum_{j=1}^{N-1}K(j,k)c_{k}^{N}-\sum_{j=1}^{\infty}K(j,k)c_{k}%
\right\vert =\left\vert \sum_{j=1}^{N_{2}}K(j,k)c_{k}^{N}-\sum_{j=1}^{N_{2}%
}K(j,k)c_{k}\right\vert +\left\vert \sum_{j=N_{2}+1}^{\infty}K(j,k)(c_{k}%
+c_{k}^{N})\right\vert .
\]

Now, the second term can be made arbitrarily small since the growth rate of
$K$ is slower than the decay of $c_{k},$ i.e., for large enough $N$%
\[
\left\vert \sum_{j=N_{2}+1}^{\infty}K(j,k)(c_{k}+c_{k}^{N})\right\vert
\leq2a(j)\frac{b(N_{2})}{N_{2}}M_{1}(0)\rightarrow0.
\]
The first term can be made as small as desired by letting $N$ grow (since
$c_{k}^{N}\rightarrow c_{k})$. Repeating the arguments of Theorem 1 we
complete the proof.
\end{proof}

We can prove a similar version of the above theorem for the symmetric kernels
assuming
\begin{equation}
(II):K(j,k)\leq a(j^{\mu})a(k^{\nu})+a(j^{\nu})a(k^{\mu})\text{ where
}a(j)=o(j)\text{ and }\mu+\nu\leq3.
\end{equation}

\begin{theorem}
Consider the EDG system given by (\ref{0-infode})-(\ref{infIC}). Let $K(j,k)$
be a nearly symmetric kernel satisfying the growth condition $(II)$ with
$\mu,\nu\leq2$ and $\mu+\nu\leq3$. Suppose for the initial distribution that
$M_{2}(0)<\infty$. Then the infinite system (\ref{0-infode})-(\ref{infIC}) has
a global solution.
\end{theorem}

\begin{proof}
The proof follows steps similar to Theorem 2. The difference is that, now, we
only have $M_{2}(0)<\infty$ for the initial distribution$.$ Since condition
$(II)$ holds, by Young's inequality one can show, as in Theorem 2,
$M_{2}(t)<\infty$ $\ $on any interval ($T<\infty)$ which allows us to
construct sequences of functions $c_{j}^{N}(t)$ which converge uniformly to
some function $c_{j}(t)$ which is continuous. To prove that $c_{j}(t)$ have
the desired properties as a solution, it is sufficient to show, arguing as in
the Theorem 3, that $\sum_{j=N}^{\infty}K(j,k)(c_{k}+c_{k}^{N})$ vanishes as
$N\rightarrow\infty$. Indeed, it is clear that$,$ for any $j\in\mathbb{N}$, we
have the bound $a(j)\leq Cj.$ Also, since $a(j)=o(j)$ we can choose $N_{2}%
\in\mathbb{N}$ large enough that, for arbitrary $\varepsilon>0,$ one has
$\frac{a(k^{\mu})}{k^{\mu}},\frac{a(k^{\nu})}{k^{\nu}}<\varepsilon$ when
$k>N_{2}.$ Then, we find
\begin{align*}
\left\vert \sum_{j=N_{2}}^{\infty}K(j,k)(c_{k}^{N}+c_{k})\right\vert  &
\leq\left\vert \sum_{j=N_{2}}^{\infty}\left(  a(j^{\mu})a(k^{\nu})+a(j^{\nu
})a(k^{\mu})\right)  (c_{k}^{N}+c_{k})\right\vert \\
&  \leq C\left\vert \sum_{j=N_{2}}^{\infty}(\varepsilon j^{2}k+\varepsilon
jk^{2})(c_{k}^{N}+c_{k})\right\vert \\
&  \leq2C\varepsilon\sup_{t\in\lbrack0,T]}M_{2}(t)M_{1}.
\end{align*}
where again we used Young's inequality in the second line. Since $\varepsilon$
is arbitrary the result follows.
\end{proof}

\section{UNIQUENESS,\ POSITIVITY AND NON-EXISTENCE}

Although the truncated system (\ref{Tode0})-(\ref{TodeIC}) has a unique
solution by the general ODE theory, the method of proof of existence we used
in the previous section does not guarantee uniqueness as there may be many
subsequences of $c_{j}^{N}$ which converges to different limit functions.
Hence, uniqueness has to be analyzed separately.

We provide two uniqueness results. Our first uniqueness result is for systems
with non-symmetric kernel. The idea is to control the "absolute" value of the
differences of two solutions, say $c_{j}$ and $d_{j},$ and show that
$c_{j}(t)=d_{j}(t)$ identically. The tricky part is the non-linear terms which
are of different signs.

\begin{theorem}
\label{uniqNsym}Consider the infinite ODE system (\ref{0-infode}%
)-(\ref{infIC}). Let the non-symmetric kernel satisfy $K(j,k)\leq Cjk$ (as in
Theorem 1). Then there is exactly one solution in $X_{2}$.
\end{theorem}

\begin{proof}
Let $c_{j}(t)$ and $d_{j}(t)$ two different solutions in $X_{2}$ with
$c_{j}(0)=d_{j}(0)$. Consider the difference $e_{j}(t)=$ $c_{j}(t)-d_{j}(t).$
Note that $\left\vert e_{j}(t)\right\vert $ is differentiable a.e. and
\[
\frac{d\left\vert e_{j}(t)\right\vert }{dt}=sgn(e_{j}(t))\frac{de_{j}(t)}%
{dt}=sgn(e_{j}(t))(\dot{c}_{j}(t)-\dot{d}_{j}(t)).
\]
Let $M_{e_{1}}(t):=\sum_{k=0}^{\infty}j\left\vert e_{j}(t)\right\vert .$ From
the rate equations, the difference $(\dot{c}_{j}(t)-\dot{d}_{j}(t))$ can be
estimated which will give terms of the form $\sum_{k=0}^{N-1}K(j,k)(c_{j}%
c_{k}-d_{j}d_{k})$ (or similarly $\sum_{k=0}^{\infty}K(k,j)(c_{j}c_{k}%
-d_{j}d_{k})).$ Consider the sum of the first $N$ terms of the series. Setting
$g_{j}=jsgn(e_{j})$, observing $(c_{j}c_{k}-d_{j}d_{k})=c_{j}e_{k}+e_{j}d_{k}$
and applying the index shifting argument as in Lemma 1 we have
\begin{align}
\sum_{j=0}^{N}jsgn(e_{j})\frac{de_{j}}{dt}  &  =\sum_{j=1}^{N}(g_{j-1}%
-g_{j})\sum_{k=0}^{\infty}K(j,k)c_{j}e_{k}+\sum_{j=1}^{N}(g_{j-1}-g_{j}%
)\sum_{k=0}^{\infty}K(j,k)e_{j}d_{k}\label{abs-der1a}\\
&  +\sum_{j=0}^{N-1}(-g_{j}+g_{j+1})\sum_{k=1}^{\infty}K(k,j)c_{j}e_{k}%
+\sum_{j=0}^{N-1}(-g_{j}+g_{j+1})\sum_{k=1}^{\infty}K(k,j)e_{j}d_{k}%
\label{abs-der1b}\\
&  +g_{N}\sum_{k=0}^{\infty}K(N+1,k)(c_{N+1}e_{k}+e_{N+1}d_{k})-g_{N}%
\sum_{k=1}^{\infty}K(k,N)(c_{N}e_{k}+e_{N}d_{k}).
\end{align}
Using the bounds on the kernel for the first terms on the right hand sides of
(\ref{abs-der1a}) and (\ref{abs-der1b}) gives%
\begin{align*}
\sum_{j=0}^{N}j\frac{d\left\vert e_{j}\right\vert }{dt}  &  \leq C\sum
_{j=1}^{N}\left\vert (g_{j-1}-g_{j})\right\vert jc_{j}M_{e_{1}}+\sum_{j=1}%
^{N}(g_{j-1}-g_{j})\sum_{k=0}^{\infty}K(j,k)e_{j}d_{k}\\
&  +C\sum_{j=0}^{N-1}\left\vert -g_{j}+g_{j+1}\right\vert jc_{j}M_{e_{1}}%
+\sum_{j=0}^{N-1}(-g_{j}+g_{j+1})\sum_{k=1}^{N}K(k,j)e_{j}^{N}d_{k}^{N}\\
&  +Cg_{N}(N+1)(c_{N+1}+d_{N+1})\sum_{k=0}^{\infty}k(c_{k}+d_{k}%
)+Cg_{N}N(c_{N}+d_{N})\sum_{k=1}^{\infty}k(c_{k}+d_{k})
\end{align*}
where we used $\left\vert e_{j}\right\vert \leq c_{j}+d_{j}$ and $c_{j}\leq C$
to simplify the sums in the last line. Since $\left\vert (g_{j-1}%
-g_{j})\right\vert \leq2j+1$ and $\left\vert -g_{j}+g_{j+1}\right\vert
\leq2j+1$ the inequality above can be written as
\begin{align*}
\dot{M}_{e_{1}}^{N}  &  \leq C\sum_{j=1}^{N}(2j+1)jc_{j}M_{e_{1}}^{N}%
+\sum_{j=1}^{N}((j-1)sgn(e_{j-1}^{N})-jsgn(e_{j}^{N}))e_{j}^{N}\sum
_{k=0}^{N-1}K(j,k)d_{k}^{N}\\
&  +C\sum_{j=0}^{N-1}(2j+1)jc_{j}M_{e_{1}}^{N}+\sum_{j=0}^{N-1}%
((j+1)sgn(e_{j+1}^{N})-jsgn(e_{j}^{N}))e_{j}^{N}\sum_{k=1}^{N}K(k,j)d_{k}%
^{N}\\
&  +Cg_{N}(N+1)(c_{N+1}+d_{N+1})M_{1}+Cg_{N}N(c_{N}+d_{N})M_{1}%
\end{align*}
The boundedness of the second moment $M_{2}(t)<\infty$ ($c,d\in X_{2})$
implies the terms in the last line vanish uniformly on bounded time intervals.
Taking the limit and noting $e_{j}^{N}=\left\vert e_{j}^{N}\right\vert
sgn(e_{j}^{N}),$ the terms on the right hand sides of first and second lines
can be bounded as
\begin{align*}
\dot{M}_{e_{1}}  &  \leq CM_{2}(t)M_{e_{1}}+\sum_{j=1}^{\infty}%
((j-1)sgn(e_{j-1})sgn(e_{j})-j)\left\vert e_{j}\right\vert \sum_{k=0}^{\infty
}K(j,k)d_{k}\\
&  +CM_{2}(t)M_{e_{1}}+\sum_{j=0}^{\infty}((j+1)sgn(e_{j+1})sgn(e_{j}%
)-j)\left\vert e_{j}\right\vert \sum_{k=1}^{\infty}K(k,j)d_{k}\\
&  \leq CM_{2}(t)M_{e_{1}}+C\sum_{j=0}^{\infty}j\left\vert e_{j}\right\vert
\sum_{k=1}^{\infty}(K(j,k)+K(k,j))d_{k}\\
&  \leq C(M_{2}(t)+M_{1})M_{e_{1}},
\end{align*}
where we used $(j-1)sgn(e_{j-1}^{N})sgn(e_{j}^{N})-j\leq0$ and
$(j+1)sgn(e_{j+1}^{N})sgn(e_{j}^{N})-j\leq1$ for the third line. Then,
applying Gronwall's lemma in the last line yields $M_{e_{1}}(t)\leq
Ce^{\int(M_{2}(s)+M_{1})ds}M_{e_{1}}(0).$ Hence we conclude $c_{j}%
(t)=d_{j}(t)$ (since $M_{e_{1}}(0)=0)$ for $j\geq1.$ To complete the proof we
also need to show $c_{0}(t)=d_{0}(t)$. Indeed, by $c_{j}(0)=d_{j}(0)$ and the
conservation of zeroth moment we have
\[
\sum_{k=1}^{\infty}c_{j}(t)=\sum_{k=1}^{\infty}c_{j}(0)=\sum_{k=1}^{\infty
}d_{j}(t).
\]
Then, since $c_{j}(t)=d_{j}(t)$ for $j\geq1$ as shown just above we
necessarily have $c_{0}(t)=d_{0}(t)$ proving uniqueness.
\end{proof}

Our second result in this section addresses the uniqueness of solutions for
symmetric kernels with faster growth. 

\begin{theorem}
\label{uniqSym} Let $c_{j}(t)$ be the solution of the system (\ref{0-infode}%
)-(\ref{infIC}) with a symmetric kernel satisyfing $K(j,k)\leq C(j^{\mu}%
k^{\nu}+j^{\nu}k^{\mu}),$ $(\mu+\nu\leq3)$ as in Theorem 2 (or $K(j,k)\leq
Cj^{2}k^{2}$ as in Corollary 3). Then there is exactly one solution to this
system in $X_{4}.$
\end{theorem}

\begin{proof}
We show the proof for kernels with the bound $K(j,k)\leq Cj^{2}k^{2},$ the
other case is similar. Let $c_{j}(t)$ and $d_{j}(t)$ be two different
solutions where again $e_{j}(t)=c_{j}(t)-d_{j}(t)$. Consider now the series
$M_{e_{2}}(t)=\sum_{j=1}^{\infty}j^{2}\left\vert e_{j}(t)\right\vert <\infty$.
Consider the difference $e_{j}(t)=$ $c_{j}(t)-d_{j}(t).$ Similar to the
uniqueness theorem for the non-symmetric (kernel) case we set $g_{j}%
=j^{2}sgn(e_{j})$. Then, by the symmetry of the kernel one can write%
\begin{align}
\sum_{j=0}^{N}j^{2}sgn(e_{j})\frac{de_{j}}{dt} &  =\sum_{j=1}^{N}%
(g_{j-1}-2g_{j}+g_{j+1})\sum_{k=0}^{\infty}K(j,k)c_{j}e_{k}\label{abs-der2a}\\
&  +\sum_{j=1}^{N}(g_{j-1}-2g_{j}+g_{j+1})\sum_{k=0}^{\infty}K(j,k)e_{j}%
d_{k}\\
&  +g_{N}\sum_{k=0}^{\infty}K(N+1,k)(c_{N+1}e_{k}+e_{N+1}d_{k})-g_{N+1}%
\sum_{k=1}^{\infty}K(k,N)(c_{N}e_{k}+e_{N}d_{k}).
\end{align}
Using the bounds on the kernel for the right hand side of (\ref{abs-der2a})
gives%
\begin{align}
\sum_{j=0}^{N}j^{2}\frac{d\left\vert e_{j}\right\vert }{dt} &  \leq
C\sum_{j=1}^{N}\left\vert (g_{j-1}-2g_{j}+g_{j+1})\right\vert j^{2}%
c_{j}M_{e_{2}}+\sum_{j=1}^{N}(g_{j-1}-2g_{j}+g_{j+1})\sum_{k=0}^{\infty
}K(j,k)e_{j}d_{k}\label{abs-der2b}\\
&  +Cg_{N}(N+1)^{2}(c_{N+1}+d_{N+1})\sum_{k=0}^{\infty}k^{2}(c_{k}%
+d_{k})+Cg_{N+1}N^{2}(c_{N}+d_{N})\sum_{k=1}^{\infty}k^{2}(c_{k}%
+d_{k})\nonumber
\end{align}
where again we used $K(j,k)\leq Cj^{2}k^{2}$ and $\left\vert e_{j}\right\vert
\leq c_{j}+d_{j}$ in the last line. Now if $c,d\in X_{4},$ the terms in the
second line of (\ref{abs-der2b}) vanishes uniformly on finite time intervals.
Also, the sums in (\ref{abs-der2a}) are bounded by $\,CM_{4}(t)M_{e_{2}}(t)$
since $\left\vert (g_{j-1}-2g_{j}+g_{j+1})\right\vert \leq Cj^{2}.$ Then
\begin{align*}
\dot{M}_{e_{2}}  & \leq CM_{4}(t)M_{e_{2}}(t)\\
& +\sum_{j=1}^{\infty}((j-1)^{2}sgn(e_{j-1}^{N})sgn(e_{j})-2j^{2}%
+(j+1)^{2}sgn(e_{j})sgn(e_{j+1}))\left\vert e_{j}\right\vert \sum
_{k=0}^{\infty}K(j,k)d_{k}%
\end{align*}
Since $(j\pm1)^{2}sgn(e_{j-1}^{N})sgn(e_{j})\leq(j\pm1)^{2}$ one gets
\begin{align*}
\dot{M}_{e_{2}} &  \leq CM_{4}(t)M_{e_{2}}(t)+\sum_{j=1}^{\infty}2\left\vert
e_{j}\right\vert \sum_{k=0}^{\infty}K(j,k)d_{k}\\
&  \leq CM_{4}(t)M_{e_{2}}(t)+CM_{e_{2}}(t)M_{2}(t)
\end{align*}
Then, by applying Gronwall's lemma one obtains $M_{e_{2}}(t)\leq
Ce^{\int(M_{4}(s)+M_{2})ds}M_{e_{2}}(0)$ which shows $c_{j}(t)=d_{j}(t)$ for
$j\geq1.$ Arguing as the previous theorem we also see $c_{0}(t)=d_{0}(t)$
completing the proof.
\end{proof}

Next we address another important property of the solutions: positivity, which
is not apparent from the equations as $\dot{c}_{j}$ terms have both,
positively and negatively signed terms. The next result guarantees this.

\begin{theorem}
\label{posFsym} Let $c_{j}(t)$ be a solution of (\ref{0-infode})-(\ref{infIC})
as in Theorem 1 (or Theorem 2). Suppose that $c_{j}(0)>0.$ Then, $c_{j}(t)>0$
for all $t>0.$
\end{theorem}

\begin{proof}
Let $S(j,c)=\sum_{k=0}^{\infty}K(j,k)c_{k},~\bar{S}(j,c)=\sum_{k=1}^{\infty
}K(k,j)c_{k}.$ Arguing as in Lemma 2, since $c_{j}(0)\geq0,$ we can easily
show that $c_{j}(t)\geq0.$ To strengthen the result, we rearrange the rate
equations and multiply the terms by the appropriate integrating factor to get
\[
\frac{d}{dt}\left[  c_{0}(t)e^{\int_{0}^{t}S(0,c(s))ds)}\right]
=c_{1}(t)S(1,c)e^{\int_{0}^{t}S(0,c(s))ds},
\]%
\[
\frac{d}{dt}\left[  c_{j}(t)e^{\int_{0}^{t}(S(j,c(s))+S(j,c(s)))ds}\right]
=(c_{j+1}(t)S(j+1,c)+c_{j-1}(t)\bar{S}(j-1,c))e^{\int_{0}^{t}%
(S(j,c(s))+S(j,c(s)))ds}.
\]

The operations on the left hand side are allowed since $S(j,c)$ and $\bar
{S}(j,c)$ are continuous by uniform convergence. Integrating this equation we
see
\begin{align*}
c_{j}(t)e^{\int_{0}^{t}(S(j,c(s))+S(j,c(s)))ds}  &  =c_{j}(0)+\int_{0}%
^{t}c_{j+1}(\tau)S(j+1,c)e^{\int_{0}^{\tau}(S(j,c(s))+S(j,c(s)))ds}d\tau\\
&  +\int_{0}^{t}c_{j-1}(\tau)\bar{S}(j-1,c)e^{\int_{0}^{\tau}%
(S(j,c(s))+S(j,c(s)))ds}d\tau.
\end{align*}
from which it follows that if $c_{j}(t)>0$\ for all $j$ since the integrals on
the right hand side are non-negative.
\end{proof}

Our final results concern the non-existence of solutions. It has been known
\cite{Dongen} and in some cases has been rigorously shown, that, for the
Smoluchowski and Becker-Doring type models, super-linearly growing kernels may
lead to non-existence \cite{Ball}, \cite{Carr}.

In EDG systems, we showed in the previous section that global solutions exist
for non-symmetric kernels satisfying $K(j,k)\leq Cjk$ and local solutions
persist for nearly symmetric kernels satisfying $K(j,k)\leq Cj^{2}k^{2}$. For
specific kernels of the form $K(j,k)=j^{\mu}k^{\nu}+j^{\nu}k^{\mu}$ ($\mu
,\nu>2)$ physical studies \cite{Naim} suggest that gelation takes place
instantaneously which is a sign of a pathological behavior. Below, taking the
approach of \cite{Ball}, we show, under some technical conditions on the
initial data and faster growth assumptions on the kernel, that the solutions
cannot exist.

To prove the result one needs to understand how the tail of the distribution
behaves with fast growing kernels. For this purpose, it will be useful to
write the infinite system as a system of density-flow equations, i.e.,
\[
\dot{c}_{j}(t)=I_{j-1}(c)-I_{j}(c),
\]
where
\begin{equation}
I_{j}(c)=c_{j}\sum_{k=1}^{\infty}K(k,j)c_{k}-c_{j+1}\sum_{k=0}^{\infty
}K(j+1,k)c_{k}. \label{flow-equ}%
\end{equation}
Again we provide two different results for the non-symmetric kernel and
symmetric kernel. For both of the results we will need the following lemma
which is a straightforward computation.

\begin{lemma}
\label{L-tail}Let $c_{j}(t)$ be a solution of the EDG system (\ref{0-infode}%
)-(\ref{infIC}). Then one has the following identities%
\begin{align*}
\sum_{j=m}^{\infty}c_{j}(t)-\sum_{j=m}^{\infty}c_{j}(0)  &  =\int_{0}%
^{t}I_{m-1}(c(s))ds,\\
\sum_{j=m}^{\infty}jc_{j}(t)-\sum_{j=m}^{\infty}jc_{j}(0)  &  =\int_{0}%
^{t}\sum_{j=m}^{\infty}I_{j}(c(s))ds+m\int_{0}^{t}I_{m-1}(c(s))ds,\\
\sum_{j=m}^{\infty}j^{2}c_{j}(t)-\sum_{j=m}^{\infty}j^{2}c_{j}(0)  &
=\int_{0}^{t}\sum_{j=m}^{\infty}(2j+1)I_{j}(c(s))ds+m^{2}\int_{0}^{t}%
I_{m-1}(c(s))ds.
\end{align*}

\end{lemma}

For the non-symmetric kernel we make the extra assumption that cluster
interaction kernels are biased, i.e., $K(k,j)>K(j,k)$ for $j>k.$ This is
reasonable assumption for systems that prefers exchanges towards bigger
clusters (e.g. migration towards bigger cities). If the exchange rate grows
faster than linearly this will cause non-existence as we see in the next theorem.

\begin{theorem}
\label{T-non1}Consider the infinite EDG system (\ref{0-infode})-(\ref{infIC})
with $c_{j}(0)>0$ for some $j$. Let $K(j,k)\geq Cj^{\beta}$ hold for some
$\beta>1.$ Assume that $K(j,0)=0,$ $K(k,j)\geq(1+\varepsilon)K(j,k)$ for
$j>k\geq1$ and some $\varepsilon>0.$ Assume further that $\lim_{m\rightarrow
\infty}e^{\delta m^{\beta-1}}\sum_{j=m}^{\infty}(j-m)c_{j}(0)\nrightarrow0$
for all $\delta>0$ (either the limit is striclty greater than zero or does not
exist). Then there exists no solution $c_{j}(t)\in X_{1}$ of (\ref{0-infode}%
)-(\ref{infIC}) on any interval $[0,T)$ $(T>0)$.
\end{theorem}

\begin{proof}
We will prove the result by contradiction. Suppose that there is a solution.
From the first and second identity of Lemma \ref{L-tail} above one has
\begin{equation}
\sum_{j=m}^{\infty}(j-m)c_{j}(t)-\sum_{j=m}^{\infty}(j-m)c_{j}(0)=\int_{0}%
^{t}\sum_{j=m}^{\infty}I_{j}(c(s))ds. \label{tail-df}%
\end{equation}
Writing in the expression for $I_{j}(c(s))$ from (\ref{flow-equ}) on the right
hand side of (\ref{tail-df}) reads%
\begin{align}
\int_{0}^{t}\sum_{j=m}^{\infty}I_{j}(c(s))ds  &  =\int_{0}^{t}\sum
_{j=m}^{\infty}c_{j}(s)\sum_{k=1}^{\infty}K(k,j)c_{k}(s)ds\label{flow-sh}\\
&  -\int_{0}^{t}\sum_{j=m}^{\infty}c_{j+1}(s)\sum_{k=0}^{\infty}%
K(j+1,k)c_{k}(s)ds. \label{nonex-sh}%
\end{align}
Shifting the index on the second term of (\ref{nonex-sh}), using $K(j,0)=0$ to
remove the $k=0$ terms and matching the lower bounds of the sums we have%
\[
\int_{0}^{t}\sum_{j=m}^{\infty}I_{j}(c(s))ds=
\]%
\[
\int_{0}^{t}\sum_{j=m}^{\infty}\sum_{k=1}^{\infty}c_{j}(s)K(k,j)c_{k}%
(s)ds-\int_{0}^{t}\sum_{j=m}^{\infty}\sum_{k=1}^{\infty}c_{j}(s)K(j,k)c_{k}%
(s)ds+\int_{0}^{t}c_{m}(s)\sum_{k=1}^{\infty}K(m,k)c_{k}(s)ds.
\]
Splitting the sums as $\sum_{j=m}^{\infty}\sum_{k=1}^{\infty}(...)=\sum
_{j=m}^{\infty}\sum_{k=1}^{m-1}(...)+\sum_{j=m}^{\infty}\sum_{k=m}^{\infty
}(...)$ and using the non-negativity of $\sum_{k=1}^{\infty}K(m,k)c_{k}(s)$
sum yields%
\begin{align*}
\int_{0}^{t}\sum_{j=m}^{\infty}I_{j}(c(s))ds  &  \geq\int_{0}^{t}\sum
_{j=m}^{\infty}\sum_{k=1}^{m-1}c_{j}(s)(K(k,j)-K(j,k))c_{k}(s)ds\\
&  +\int_{0}^{t}\sum_{j=m}^{\infty}\sum_{k=m}^{\infty}c_{j}%
(s)(K(k,j)-K(j,k))c_{k}(s)ds.
\end{align*}
Note that the second double-sum on the right hand side is zero by the symmetry
of the sum and hence by (\ref{tail-df}) we are left with
\begin{align}
\sum_{j=m}^{\infty}(j-m)c_{j}(t)-\sum_{j=m}^{\infty}(j-m)c_{j}(0)  &  \geq
\int_{0}^{t}\sum_{j=m}^{\infty}\sum_{k=1}^{m-1}c_{j}(s)(K(k,j)-K(j,k))c_{k}%
(s)ds\\
&  \geq\int_{0}^{t}\sum_{j=m}^{\infty}\sum_{k=1}^{m-1}\varepsilon
c_{j}(s)K(j,k)c_{k}(s)ds.
\end{align}
Using the lower bound $K(j,k)\geq Cj^{\beta}$ one has
\begin{align}
\sum_{j=m}^{\infty}(j-m)c_{j}(t)  &  \geq\sum_{j=m}^{\infty}(j-m)c_{j}%
(0)+\int_{0}^{t}\varepsilon C\sum_{j=m}^{\infty}j^{\beta}c_{j}(s)\sum
_{k=1}^{m-1}c_{k}(s)ds\\
&  \geq\sum_{j=m}^{\infty}(j-m)c_{j}(0)+\varepsilon Cm^{\beta-1}\int_{0}%
^{t}\sum_{j=m}^{\infty}jc_{j}(s)ds. \label{yokluk-ineq}%
\end{align}
For the second line we used $\sum_{k=1}^{m-1}c_{k}(s)\geq C>0$ which is a
consequence of the fact $c_{j}(0)>0$ for some $j$ which is aconsequence of the
continuity of solutions and the fact that the system has non-zero initial mass
(and hence the finite sum is strictly greater than zero for large enough $m$).
Since $\sum_{j=m}^{\infty}jc_{j}(t)>\sum_{j=m}^{\infty}(j-m)c_{j}(t),$ one
gets the differential inequality below%
\[
\sum_{j=m}^{\infty}jc_{j}(t)\geq\sum_{j=m}^{\infty}(j-m)c_{j}(0)+\varepsilon
Cm^{\beta-1}\int_{0}^{t}\sum_{j=m}^{\infty}jc_{j}(s)ds,
\]
from which we get the inequality%
\[
\sum_{j=m}^{\infty}jc_{j}(t)\geq e^{\varepsilon Cm^{\beta-1}t}\sum
_{j=m}^{\infty}(j-m)c_{j}(0).
\]
Since $\lim_{m\rightarrow\infty}e^{\varepsilon Cm^{\beta-1}t}\sum
_{j=m}^{\infty}(j-m)c_{j}(0)\nrightarrow0$ for any $t>0$ by our assumption we
arrive at $\lim_{m\rightarrow\infty}\sum_{j=m}^{\infty}jc_{j}(t)>0$ which is a contradiction.
\end{proof}

\textbf{Example:} The condition on $c_{j}(0)$ in the above theorem can be
achieved by many kinds of initial distributions with algebraically decaying
tails. Consider, for instance, $c_{j}(0)=\frac{1}{j^{q}}$ with any $q>1.$
Then,
\[
\sum_{j=m}^{\infty}(j-m)c_{j}(0)\geq\sum_{j=m+1}^{\infty}c_{j}(0)=\sum
_{j=m+1}^{\infty}\frac{1}{j^{q}}%
\]
Comparing the sum $\sum_{j=m+1}^{\infty}\frac{1}{j^{q}}$ with the integral
$\int_{m+1}^{\infty}\frac{dy}{y^{q}}$ we obtain
\[
\sum_{j=m}^{\infty}(j-m)c_{j}(0)\geq\frac{1}{(q-1)(m+1)^{q-1}}>0.
\]
Then in the limit $m\rightarrow\infty$ the condition of the theorem is
satisfied for any $q>1$ since%
\[
\lim_{m\rightarrow\infty}e^{\delta m^{\beta-1}}\left(  \frac{1}%
{(q-1)(m+1)^{q-1}}\right)  >0.
\]
If we assume faster growth such as $\beta>2,$ then the condition in the
theorem is satisfied even by distributions with light tails. Indeed, let
$c_{j}(0)=\kappa^{j},$ $\kappa<1.$ Then, $\sum_{j=m}^{\infty}(j-m)c_{j}%
(0)=\frac{\kappa^{m+1}}{(1-\kappa)^{2}}.$ Hence one has, for any $\delta,$
\[
\lim_{m\rightarrow\infty}e^{\delta m^{\beta-1}}\sum_{j=m}^{\infty}%
(j-m)c_{j}(0)=\lim_{m\rightarrow\infty}e^{\delta m^{\beta-1}-C-m\ln
(\kappa)+\ln(\kappa/(1-\kappa)^{2})}>0
\]
satisfying the condition of the theorem.

The previous theorem relied on the assumption that pairwise interactions
favored bigger sizes. For symmetric kernels, there is no such favoring and
non-existence cannot take place unless $K(j,k)$ grows faster (agreeing with
the existence results of the previous section). However, we have the following result.

\begin{theorem}
Consider the infinite EDG system (\ref{0-infode})-(\ref{infIC}). Suppose that
the symmetric kernel satisfies $K(j,k)\geq Cj^{\beta}$ for some $\beta>2.$
Assume also that $\lim_{m\rightarrow\infty}e^{\delta m^{\beta-2}}\sum
_{j=m}^{\infty}(j^{2}-m^{2})c_{j}(0)\nrightarrow0$ for all $\delta>0$. Then
there exists no solution $c_{j}(t)\in X_{2}$ of (\ref{0-infode})-(\ref{infIC})
on any interval $[0,T)$ $(T>0)$.
\end{theorem}

\begin{proof}
We go by contradiction as in Theorem \ref{T-non1}. Let $c_{j}(t)\in X_{2}$ be
a solution on $[0,T).$ Then $M_{2}(t)<\infty$ for $t<T.$ Using the first and
third identities of Lemma \ref{L-tail} we have
\begin{equation}
\sum_{j=m}^{\infty}(j^{2}-m^{2})c_{j}(t)-\sum_{j=m}^{\infty}(j^{2}-m^{2}%
)c_{j}(0)=\int_{0}^{t}\sum_{j=m}^{\infty}(2j+1)I_{j}(c(s))ds.\label{tail-df2}%
\end{equation}
Pulling $I_{j}(c(s))$ from (\ref{flow-equ}) and placing it on the right hand
side of (\ref{tail-df2}) and shifting the index for the $c_{j+1}$ term reads%
\begin{align}
\int_{0}^{t}\sum_{j=m}^{\infty}(2j+1)I_{j}(c(s))ds &  =\int_{0}^{t}\sum
_{j=m}^{\infty}(2j+1)c_{j}(s)\sum_{k=1}^{\infty}K(k,j)c_{k}%
(s)ds\label{flow-sh2}\\
&  -\int_{0}^{t}\sum_{j=m+1}^{\infty}(2j-1)c_{j}(s)\sum_{k=0}^{\infty
}K(j,k)c_{k}(s)ds.\label{nonex-sh2}%
\end{align}
Matching the lower indices in (\ref{flow-sh2}), (\ref{nonex-sh2}) for the $j$
sums and removing the $k=0$ terms in (\ref{nonex-sh2}) by symmetry ($K(0,j)=0$
identically$)$ one gets the inequality%
\[
\int_{0}^{t}\sum_{j=m}^{\infty}(2j+1)I_{j}(c(s))ds\geq
\]%
\[
2\int_{0}^{t}\sum_{j=m}^{\infty}c_{j}(s)\sum_{k=1}^{\infty}K(k,j)c_{k}%
(s)+\int_{0}^{t}\sum_{j=m}^{\infty}(2j-1)c_{j}(s)\sum_{k=1}^{\infty
}(K(k,j)-K(j,k))c_{k}(s)ds.
\]
where we used the non-negativity of $\int_{0}^{t}(2m-1)c_{j}(s)\sum
_{k=0}^{\infty}K(m,k)c_{k}(s)ds.$ Notice, by symmetry, the second term in the
second line is zero. Then, placing the remaining inequality in equation
(\ref{tail-df2}) we see%
\begin{equation}
\sum_{j=m}^{\infty}(j^{2}-m^{2})c_{j}(t)-\sum_{j=m}^{\infty}(j^{2}-m^{2}%
)c_{j}(0)\geq2\int_{0}^{t}\sum_{j=m}^{\infty}c_{j}(s)\sum_{k=1}^{\infty
}K(k,j)c_{k}(s)ds.\label{mom-ters}%
\end{equation}
Now, by the bounds for $K$ assumed in the theorem, we can write, from
(\ref{mom-ters}), the following%
\begin{align*}
\sum_{j=m}^{\infty}j^{2}c_{j}(t) &  \geq\sum_{j=m}^{\infty}(j^{2}-m^{2}%
)c_{j}(0)+2C\int_{0}^{t}\sum_{j=m}^{\infty}j^{\beta}c_{j}(s)\sum_{k=1}%
^{\infty}c_{k}(s)ds\\
&  \geq\sum_{j=m}^{\infty}(j^{2}-m^{2})c_{j}(0)+2Cm^{\beta-2}\int_{0}^{t}%
\sum_{j=m}^{\infty}j^{2}c_{j}(s)ds.
\end{align*}
where in the second line we used $\sum_{k=1}^{\infty}c_{k}(s)\geq C>0$ as in
Theorem 8. Solving the differential inequality yields the inequality%
\[
\sum_{j=m}^{\infty}j^{2}c_{j}(t)\geq e^{2Cm^{\beta-2}t}\sum_{j=m}^{\infty
}(j^{2}-m^{2})c_{j}(0)
\]
which contradicts, in the limit $m\rightarrow\infty,$ with the boundedness of
$M_{2}(t)$ on finite intervals$.$
\end{proof}

\section{CONCLUSION}

In this article, as an initial mathematical investigation of the subject, we
studied fundamental properties of the EDG systems. For the last two decades,
EDG type models have attracted considerable attention of the interdisciplinary
communities as such models have found applications in physics, migration
dynamics, socioeconomic behavior etc. Also, related particle level processes
(e.g. zero-range processes) are also of significant interest as the rate
equations that we studied in this article can be obtained as limits of
underlying stochastic dynamics. With the development of the subject in
multiple avenues two different but related views of the mass exchange
processes grew around physics and probability fields. Our article is motivated
by the latter approach.

The connection between the two approaches is simple but subtle. For a
physicist the exchange processes are meaningfully defined only between
clusters that have non-zero mass and growth is unidirectional. So, when a
monomer is absorbed into another cluster there remains nothing behind. In the
course of the time the total mass $\sum_{j=1}^{\infty}jc_{j}(t)$ is the only
conserved quantity and total number of clusters $\sum_{j=1}^{\infty}c_{j}(t)$
decreases in time. In the probabilists' view the particles sit on lattice
sites (or on a complete graph) each of which can accommodate arbitrary number
of particles. Masses on the lattice sites interact with each other in a
similar way that clusters interact in the physicists' picture, that is, by
exchanging particles among each other one at a time. There is one significant
difference however, namely the `empty sites' or `empty (available) volume'. In
our formulation, which is the more general one, particles are allowed to hop
from a massive cluster to empty (available) volume creating a single monomer
which can continue to interact with the rest of the system in the usual way.
And when a monomer is taken by another cluster the remaining space is still
available to be occupied. In this regard, the `total volume' or total number
of clusters including the zero-cluster (or the available volume), i.e.,
$\sum_{j=0}^{\infty}c_{j}(t)$ is conserved. These two views are compatible
with each other and in fact one can be "obtained" from the other. By setting
$K(j,0)=0$ in our general formulation, we disallow hopping to the available
volume and system grows indefinitely creating more and more available volume
in time. Indeed, looking at the rate equations (\ref{0-infode})-(\ref{infIC}),
if $K(j,0)=0,$ we observe that $c_{0}(t)$ monotonically increases which means
that $\sum_{j=1}^{\infty}c_{j}(t)$ must decrease due to conservation of total
volume just as a physicist would reason. We also observe that the rate
equations for $j\geq1$ is completely decoupled from the $c_{0}(t)$ and evolve
independently again agreeing with physicists' picture of the process. However,
the main theorems on the existence and uniqueness that are proven in this
article remain intact and give us all the existence and uniqueness results for
the classical EDG system (after choosing a "free" initial condition for
$c_{0}).$

To recapitulate our results, we showed that growth assumptions on the kernel
determine whether the solutions exist globally, locally or do not exist at
all. In particular, for general non-symmetric kernels whose growth is bounded
as $K(j,k)\leq Cjk,$ unique classical solutions exist globally. For symmetric
kernels however, we showed that the existence result can be generalized to
kernels whose growth rate is lying in the range $K(j,k)\leq C(j^{\mu}%
k^{v}+j^{\nu}k^{\mu}),$ with $\mu,\nu\leq2,$ $\mu+v\leq3$. This fact was first
discovered by physicists based on scaling arguments \cite{Naim}. On the other
hand, for non-symmetric kernels which grow fast enough (i.e., $K(j,k)\geq
Cj^{\beta})$ ($\beta>1)$ we showed that the solutions can not exist at all.
Similarly, for symmetric kernels, we proved an analogous results stating that,
for kernels which grow with the rate $K(j,k)\geq Cj^{\beta})$ ($\beta>2)$
solutions cease to exist if some assumptions on the initial conditions are satisfied.

A number of questions remain still open for investigation. First of all, the
intriguing question of existence of gelling solutions (solutions that do not
conserve mass) is not addressed in this article. Physical studies suggest that
$\mu=3/2$ is the critical exponent beyond which gelation takes place. A
separate but related question in this matter is whether the gelling solutions
(if they exist) can be extended beyond the gelation time. Also, physical
studies suggest that for kernels that grow super-quadratically, gelation takes
place instantaneously for general initial conditions. Although, our
non-existence result is a step in that direction, it is by no means a complete
resolution of the problem as we restricted ourselves to specific initial conditions.

Another whole area which deserves detailed analysis and which we have made no
attempt to analyze is the existence of self-similar solutions and large time
behavior of general solutions. In recent years there has been revived interest
on the subject and several seminal results has been obtained for Smoluchowski
type models concerning self-similarity \cite{Laurencot}, \cite{Esco2} and the
long time behaviour. Similar results are likely to be true for the case of the
EDG systems and have been considered by physicists for kernels with special
form \cite{Ke1}. Yet another interesting line of research direction is the
investigation of equilibria, their existence, convergence of general solutions
to the equilibrium and the possibility of dynamics phase transitions and its
relation to the condensation phenomena that appear in zero-range processes.

\textbf{Acknowledgements}. I thank Colm Connaughton and Stefan Grosskinsky for
fruitful discussions. The author is supported by the Marie Curie Fellowship of
European Commission, grant agreement number: 705033.

\end{document}